\documentclass[a4paper,12pt]{article}
\usepackage{amsmath}
\usepackage{a4wide}
\usepackage{amsfonts}
\usepackage{amssymb,,amsmath}
\usepackage{graphicx}
\usepackage{fullpage}
\usepackage{graphicx}
\usepackage{amsmath}
\usepackage{amsthm} 
\usepackage{dsfont}
\usepackage{amsfonts}
\usepackage{amssymb}
\usepackage{array}
\usepackage{multicol}
\usepackage[latin1]{inputenc}
\usepackage[T1]{fontenc}
\usepackage{MnSymbol}
\usepackage{todonotes}
\usepackage{braket}
\usepackage{cite}

\theoremstyle{plain}
\newtheorem{de}{Definition}[section]
\newtheorem{theo}{Theorem} [section] 
\newtheorem{prop}{Proposition}[section]

\theoremstyle{definition}
\newtheorem{ex}{Examples}[section]
\theoremstyle{remark}
\newtheorem{rem}{Remarks}[section]

\newcommand{\Hcal}{\mathcal{H}}
\newcommand{\Kcal}{\mathcal{K}}
\newcommand{\Ncal}{\mathcal{N}}
\newcommand{\Lcal}{\mathcal{L}}

\newcommand{\Pcal}{\mathcal{P}}
\newcommand{\Acal}{\mathcal{A}}

\newcommand{\Bcal}{\mathcal{B}}

\newcommand{\Mcal}{\mathcal{M}}

\newcommand{\Prob}{\mathbb{P}}
\newcommand{\Pbb}{\mathbb{P}}

\newcommand{\Wbb}{\mathbb{W}}

\newcommand{\R}{\mathbb{R}}

\newcommand{\Lbb}{\mathbb{L}}

\newcommand{\Sbb}{\mathbb{S}}
\newcommand{\C}{\mathbb{C}}

\newcommand{\Ker}{\text{Ker}\ }
\newcommand{\Ima}{\text{Im}\ }

\newcommand{\sca}[2]{\langle #1 ,#2\rangle}

\newcommand{\proj}[2]{|#1\rangle\langle #2|}

\newcommand{\Tr}{\mbox{\rm Tr}}
\def\indic{{\mathop{\rm 1\mkern-4mu l}}}
\font\timesept=cmr7

\numberwithin{equation}{section}

\begin{document}
\title{Classical and Quantum Part of the Environment\\ for Quantum Langevin Equations\thanks{Work supported by A.N.R. grant: ANR-14-CE25-0003 "StoQ" }}
\author{St\'ephane Attal and Ivan Bardet}
\date{}
\maketitle
\vskip -8mm
\centerline{\timesept 
Institut Camille Jordan}
\vskip -1mm
\centerline{\timesept Universit\'e Claude Bernard Lyon 1}
\vskip -1mm
\centerline{\timesept 43 bd du 11 novembre 1918}
\vskip -1mm
\centerline{\timesept 69622 Villeubanne cedex, France
}

\bigskip
\begin{abstract}
 Among quantum Langevin equations describing the unitary time evolution of a quantum system in contact with a quantum bath, we completely characterize those equations which are actually driven by classical noises. The characterization is purely algebraic, in terms of the coefficients of the equation. In a second part, we consider general quantum Langevin equations and we prove that they can always be split into a maximal part driven by classical noises and a purely quantum one. 
\end{abstract}

\section{Introduction}\label{intro}

Since the construction of Quantum Stochastic Calculus and the corresponding quantum stochastic differential equations (quantum Langevin Equations) on the symmetric Fock space (\cite{H-P}), it is well-known that both classical and quantum noises could coexist in the equation. The framework is designed for quantum noises, however some classical noises can also appear with some particular combinations of the quantum noises (see S. Attal's lecture in \cite{AJP2}). This fact was the starting point of the recent articles \cite{ADP2} and \cite{ADP2016}, where the authors characterized all the possible classical processes that can emerge in such quantum Langevin equations: they are the complex normal martingales in $\C^n$; up to a unitary transform of $\C^n$ they are combinations of independent Wiener processes and Poisson processes in different directions of the space.

Let us be more explicit with some simple examples. The quantum system state space is the separable Hilbert space $\Hcal$, possibly infinite dimensional, whereas the quantum heat bath is represented by quantum noises $da^i_j(t)$ (see Subsection \ref{notations} for the notations) on the symmetric Fock space $\Phi=\Gamma_s(L^2(\R^+; \C))$. In the simplest case where $n=1$, the joint evolution between the system and its environment can be described by a one-parameter family of unitary operators $(U_t)$, solving a quantum Langevin equation:
\begin{equation}\label{eqHPintro}
dU_t=-\left(iH+\frac{1}{2}L^*L\right)\,U_t\,dt+L\,U_t\, da^0_1(t)-L^*S\, U_t\,da^1_0(t)+(S-I)\,U_t\,da^1_1(t)\,.
\end{equation}
where $H,L$ and $S$ are operators on $\Hcal$ such that $H=H^*$ and $S$ is a unitary operator. 

In the particular case where $S=I$ and $L=-L^*$, Equation \eqref{eqHPintro} becomes
\begin{equation}\label{eqHPintrobruitbrownian}
dU_t=-\left(iH+\frac{1}{2}L^2\right)\,U_t\,dt+L\,U_t\,\left(da^0_1(t)+da^1_0(t)\right)\,.
\end{equation}
But it is well-known that the combinations of operators $B_t=a^0_1(t)+a^1_0(t)$ are naturally isomorphic to the multiplication operators by the Brownian motion on its canonical space. Hence Equation \eqref{eqHPintrobruitbrownian} is actually a Brownian motion driven unitary evolution:
$$
dU_t=-\left(iH+\frac{1}{2}L^2\right)\,U_t\,dt+L\,U_t\,dB_t\,.
$$
Note that the conditions on $H$ and $S$ are the most general ones for a Brownian motion driven operator-valued equation to give unitary solutions. 

In the other particular case where $L=\rho\,(S-I)$, for some $\rho>0$, Equation \eqref{eqHPintro} becomes
\begin{equation}\label{eqHPintrobruitpoisson}
dU_t=\left(iH-\frac{1}{2}L^*L\right)\,U_t\,dt+L\,U_t\, \left(da^0_1(t)+da^1_0(t)+\frac 1\rho\,da^1_1(t)\right)\,.
\end{equation}
The combinations of operators $X_t=a^0_1(t)+a^1_0(t)+\left(1/\rho\right)\,\, a^1_1(t)$ are naturally isomorphic to the multiplication operators by the compensated Poisson process with intensity $\rho^2$ and jumps $1/\rho$ on its canonical space. Hence Equation \eqref{eqHPintrobruitpoisson} is actually a Poisson process driven unitary evolution:
\begin{equation*}
dU_t=-\left(iH-\frac{1}{2}L^*L\right)\,U_t\,dt+L\,U_t\, dX_t\,.
\end{equation*}
Note that the conditions above on the coefficients are the most general ones for a Poisson process driven operator-valued equation to give unitary solutions. 

\smallskip
In general, it can be shown that for the von Neumann algebra generated by $\Bcal(\Hcal)\otimes I_\Phi$ and $\{U_t\}_{t\geq0}$ to be of the form $\Bcal(\Hcal)\otimes\Acal$ with $\Acal$ commutative, one of the two following conditions must holds:

\smallskip\noindent
-- either $S=I_\Hcal$ and there exists $\theta\in\R$ such that $L^*=e^{i\theta}L$,

\smallskip\noindent
-- or there exists a complex number $\lambda$ such that $L=\lambda(S-I_\Hcal)$.

\smallskip
The first motivation of this paper is to give a criteria on the unitary solution of a quantum Langevin Equation so that it is driven by classical noises, generalizing the preceding remark.

On the other hand, some evolutions are understood to be typically quantum or non-commutative, although there is no clear definition of what it means. This is for instance the case for the spontaneous emission, where the evolution is given by the unitary solution of the following quantum Langevin Equation:

\begin{equation}\label{eqspontem}
dU_t=-\frac{1}{2}V^*V\,U_t\,dt+V\, U_t\,da^0_1(t)-V^*\,U_t\,da^1_0(t),
\end{equation}
where 
$$
V=\begin{pmatrix}0&1\\0&0\end{pmatrix}\,.
$$
In a second part of the article we show that any quantum Langevin equation can be split into two parts: a maximal commutative one (that is, driven by a classical noise and maximal in dimension) and a purely quantum one (that is, which contains no classical part whatsoever). 

\smallskip
The article is structured as follows. In Section 2 we recall a few notations concerning quantum noises, quantum Langevin equations and their probabilistic interpretations. We discuss the notion of \emph{change of noise}, that is, the effect on the quantum noise of a change of basis. We recall the definition and the main properties of the \emph{noise algebra}, as defined in \cite{Bar2}.

Section 3 is devoted to our main result: a complete characterization of those quantum Langevin equations that give rise to a commutative noise algebra. The characterization is given both in algebraic properties of the coefficients and in probabilistic interpretations of the classical noises appearing in the equation. As a corollary we obtain a characterization of the corresponding Lindblad generators. 

In Section 4, we are back to general quantum Langevin equations and we prove that they all admit a splitting into a maximal commutative part and a purely quantum one (in the sense that it admits no commutative subspace). We end this section and the article with some discussion and examples.

\section{The Noise Algebra}

\subsection{Notations}\label{notations}

Let us recall here a few notations concerning quantum noises.

\smallskip
Let $\Kcal$ be a \emph{finite} dimensional Hilbert space of dimension $d$. We put $\Lambda=\{1,...,d\}$ and we consider a fixed orthonormal basis $(e_i)_{i\in\Lambda}$ of $\Kcal$. We denote by $\Phi=\Gamma_s\left(L^2(\R^+,\Kcal)\right)$ the symmetric Fock space over $L^2(\R^+,\Kcal)$. We are given ourselves another auxiliary separable Hilbert space $\Hcal$, possibly infinite dimensional, which represents the "small system" state space in quantum Langevin equations. We put $\Psi=\Hcal\otimes\Phi$. 

On the Fock space $\Phi$ we consider the usual creation operators $A^\dagger(f)$ and annihilation operators $A(f)$, for all $f\in L^2(\R^+,\Kcal)$; we also consider the differential second quantization operators $d\Gamma(H)$, for all self-adjoint operator $H$ on $L^2(\R^+,\Kcal)$. The quantum noises $a^i_j(t)$, $i,j\in\Lambda\cup\{0\}$, are then defined as follows:
\begin{align*}
a^0_i(t)&=A^\dagger\left(\indic_{[0,t]}\,\vert e_i\rangle\right)\,,\\
a^i_0(t)&=A\left(\indic_{[0,t]}\,\langle e_i\vert\right)\\
a^i_j(t)&=d\Gamma\left(\vert e_j\rangle\langle e_i\vert\Mcal_{\indic_{[0,t]}}\right)\,,
\end{align*}
where $\indic_{[0,t]}$ is the the usual indicator function of the interval $[0,t]$, where $\Mcal_{\indic_{[0,t]}}$ is the multiplication operator by $\indic_{[0,t]}$ and where we used the usual ``bra" and ``ket" notations for vectors and linear forms on $\Kcal$. 

\smallskip
Quantum noises are driving quantum Langevin equations of all sorts. But it is a remarkable fact that some particular combinations of the quantum noises actually represent well-known classical noises. It can be shown (see S. Attal's lecture in \cite{AJP2}) that:

\smallskip\noindent
-- the operators $a^0_i(t)+a^i_0(t)$ are naturally isomorphic to the multiplication operators by independent Brownian motions $B^i_t$ acting on their canonical spaces,

\smallskip\noindent
-- the operators $a^0_i(t)+a^i_0(t)+\lambda a^i_i(t)$ are naturally isomorphic to the multiplication operators by independent compensated Poisson processes $X^i_t$, with jumps $\lambda$ and intensity $1/\lambda^2$,  acting on their canonical spaces. 

\smallskip
In some part of the proof of the main theorem (Theorem \ref{theocomnoisealgebra}) we shall meet, in one particular case, the operator process $(a^i_i(t))_{t\in\R^+}$ alone. Though it is a commutative family of self-adjoint operators and as such they are unitarily equivalent to multiplication operators by a classical process, they are of deterministic law $\delta_0$ in the reference state of the Fock space (the vacuum state). Hence they bring nothing to the probabilistic interpretation of the associated equation, nor to the associated Lindblad generator. They are of no effect on the small system. 

\subsection{Quantum Langevin Equations}

Now we recall a few elements of unitary quantum Langevin equations, in the framework of Hudson-Parthasarathy Quantum Stochastic calculus \cite{H-P} \cite{Mey2}.

\smallskip
On the space $\Psi$ we consider the following quantum stochastic equation
\begin{equation}
dU_t=\sum_{i,j\in\Lambda\cup\{0\}}L^i_j\,U_t\,da_j^i(t)\,,
\label{eqHPeq}
\end{equation}
where the $a^i_j$'s are the quantum noises on $\Phi$ and where the $L^i_j$'s are operators on $\Hcal$. The following well-known theorem characterizes in terms of the operators $L^i_j$ the fact that the solution $(U_t)$ is made of unitary operators or not.

\begin{theo}[\cite{H-P}]\label{theoHPeq}
If the  $L_j^i$'s are bounded operator on $\Hcal$, then Equation \eqref{eqHPeq} admits a unique solution on $\Bcal(\Psi)$.
\\Furthermore, this solution is made of unitary operators $U_t$ if and only if there exist bounded operators $H$ and $S^i_j$ ($i,j\in\Lambda$) on $\Hcal$ such that

\smallskip\noindent
i) the operator $H$ is selfadjoint,

\smallskip\noindent
ii) the operator $\Sbb=\sum_{i,j\in\Lambda}{S^i_j\otimes\proj{j}{i}}$ on $\Bcal(\Hcal\otimes\Kcal)$ is unitary

\smallskip\noindent
iii) the coefficients $L^i_j$ are of the form:
\begin{align*}
L_0^0&= -iH-\frac{1}{2}\sum_{k\in\Lambda}{\left(L^0_k\right)^*L^0_k}\\
L^i_0&=-\sum_{j\in\Lambda}{\left(L^0_j\right)^*S^i_j} \\
L^i_j&=S^i_j-\delta_{i,j}I_\Hcal\,,
\end{align*}
for all $ i,j\in\Lambda$. 
\end{theo}

\subsection{Change of noises}

Usually the orthonormal basis $(e_i)_{i\in\Lambda}$ is given by the context and one does not change it, once it is fixed. It is clear that the choice of this basis determines the coefficients taking part in Equation \eqref{eqHPeq}. As an example, consider again the quantum Langevin equation \eqref{eqHPintro} given in the introduction, with $S=I_\Hcal$:
\begin{equation*}
dU_t=-\left(iH+\frac{1}{2}L^*L\right)U_t\,dt+L\,U_t\,da^0_1(t)-L^*\,U_t\,da^1_0(t)\,.
\end{equation*}
Here $K=\C e_1$ is one dimensional; one would think that the choice of the basis, i.e. the choice of the unit vector $e_1\in\Kcal$ is not important. However, suppose that ${L}^*=\lambda\,L$ for some $\lambda\in\C$, $\vert\lambda\vert=1$ (this situation may happen whenever the Noise Algebra is commutative, as we shall see in Section \ref{commEnv}), the previous equation becomes
\begin{equation}\label{eitheta}
dU_t=-\left(iH+\frac{1}{2}L^2\right)\,U_t\,dt+L\,U_t\,da^0_1(t)-\lambda\,L\,U_t\,da^1_0(t)
\end{equation}
With this choice of a basis it is not clear that the equation is actually driven by a classical noise. However, take $\mu\in \C$ such that $\mu^2=-\lambda$. If one takes as a new basis the vector $f_1=\overline{\mu}\,e_1$, then 
$$
da^0_1(t)=dA^\dagger\left(\vert e_1\rangle\,\indic_{[0,t]}\right)=dA^\dagger\left(\vert \mu\,f_1\rangle\,\indic_{[0,t]}\right)=\mu\,dA^\dagger\left(\vert f_1\rangle\,\indic_{[0,t]}\right)=\mu\,d\tilde{a}^0_1(t)\,.
$$
On the other hand, as $da^1_0(t)$ is the adjoint of $da^0_1(t)$, we get
$$
da^1_0(t)=\overline{\mu}\,d\tilde{a}^1_0(t)\,.
$$
Hence Equation \eqref{eitheta} becomes
$$
dU_t=-\left(iH+\frac{1}{2}L^2\right)\,U_t\,dt+\mu\,L\,U_t\,\left(d\tilde a^0_1(t)+\,d\tilde a^1_0(t)\right)\,.
$$
We recognize a usual Brownian motion driven quantum Langevin equation
$$
dU_t=-\left(iH+\frac{1}{2}\tilde L^2\right)\,U_t\,dt+\tilde{L}\,U_t\,dB_t,
$$
where $B_t=\tilde{a}^0_1(t)+\tilde{a}^1_0(t)$ is unitarily isomorphic to the multiplication operator by a real Brownian motion and where we put $\tilde{L}=\mu L$, so that ${\tilde{L}}^*=-\tilde L$. 

It is now clear that in order to unravel classical noises in quantum Langevin equations we must allow changes of basis in $\Kcal$. We shall call such a transformation a \emph{change of noises}.

In order to make the following more readable we fix the following notations. 
Consider a quantum Langevin equation on $\Psi$ of the form
\begin{equation*}
dU_t=L^0_0\, U_t\, dt+\sum_{i=1}^d L^0_i\, U_t\, da^0_i(t)+\sum_{i=1}^d L^i_0\,U_t\, da^i_0(t)+\sum_{i,j=1}^d L^i_j U_t\, da^i_j(t)\,,
\end{equation*}
It will be convenient in the sequel to consider the coefficients $L^0_i$ as a column vector
$$
L^0=\begin{pmatrix} L^0_1\\\vdots\\L^0_d\end{pmatrix}\,,
$$
the coefficients $L^i_0$ as a row vector
$$
L_0=\begin{pmatrix} L^1_0&\ldots&L^d_0\end{pmatrix}
$$
and the coefficients $L^i_j$ as a $d\times d$-block-matrix $\Lbb$ such that $\Lbb_{ij}=L^j_i$.

Note that, consistently with these notations, we have
$$
\left(L^0\right)^*=\begin{pmatrix} \left(L^0_1\right)^*&\ldots&\left(L^0_d\right)^*\end{pmatrix}\,.
$$

\begin{prop}\label{change}
Consider a quantum Langevin equation on $\Psi$ of the form
\begin{equation}\label{QLE1}
dU_t=L^0_0\, U_t\, dt+\sum_{i=1}^d L^0_i\, U_t\, da^0_i(t)+\sum_{i=1}^d L^i_0\,U_t\, da^i_0(t)+\sum_{i,j=1}^d L^i_j U_t\, da^i_j(t)\,,
\end{equation}
where the quantum noises $a^i_j$ are associated to a given orthonormal basis $(e_i)_{i\in\Lambda}$ of $\Kcal$. 

In the orthonormal basis $(f_i)_{i\in\Lambda}$ of $\Kcal$, given by $f_i=We_i$, $i\in\Lambda$, for some unitary operator $W$ on $\Kcal$, Equation \eqref{QLE1} becomes
\begin{equation}\label{QLE2}
dU_t=L^0_0\, U_t\, dt+\sum_{i=1}^d \tilde L^0_i\, U_t\, d\tilde a^0_i(t)+\sum_{i=1}^d \tilde L^i_0\,U_t\, d\tilde a^i_0(t)+\sum_{i,j=1}^d \tilde L^i_j U_t\, d\tilde a^i_j(t)\,,
\end{equation}
where
\begin{align*}
\tilde L^0&=W^*\, L^0\,,\\
\tilde L_0&=L_0\, W\,,\\
\tilde \Lbb&=W^*\, \Lbb\, W\,.
\end{align*}
\end{prop}
\begin{proof}
We have $e_i=\sum_{j\in\Lambda} (W^*)_{ji}\, f_j$, so that 
\begin{align*}
da^0_i(t)&=\sum_{j\in\Lambda} (W^*)_{ji}\,d\tilde a^0_j(t)\\
da^i_0(t)&=\sum_{j\in\Lambda} W_{ij}\, d\tilde a^j_0(t)\\
da^i_j(t)&=\sum_{k,l\in\Lambda} W_{ik}\,(W^*)_{lj}\, d\tilde a^k_l(t)\,.
\end{align*}
This gives
\begin{multline*}
dU_t=L^0_0\, U_t\, dt+\sum_{j=1}^d\left(\sum_{i=1}^d (W^*)_{ji}\,L^0_i\right)\, U_t\, d\tilde a^0_j(t)+\sum_{j=1}^d\left(\sum_{i=1}^d L^i_0 W_{ij}\right)\,U_t\, d\tilde a^i_0(t)+\hfill\\
\hfill+\sum_{k,l=1}^d\left(\sum_{i,j=1}^d (W^*)_{lj}\, \Lbb_{ji}\, W_{ik}\right)\, U_t\, d\tilde a^k_l(t)\,.
\end{multline*}
This gives the result.
\end{proof}

\smallskip
Regarding the case of unitary-valued quantum Langevin equations, the proposition above shows that the conditions on the $L^i_j$'s are not affected by changes of noise, as is summarized below. 

\begin{prop}
Consider a unitary-valued quantum Langevin equation of the form
\begin{multline}\label{eqHPequnitayform}
dU_t=-\left(i\,H+\frac{1}{2}\sum_{k\in\Lambda}{\left(L^0_k\right)^*\,L^0_k}\right)\,U_t\,dt+\sum_{k\in\Lambda}{L^0_k\,U_t\,da^0_k(t)}\hfill\\
\hfill +\sum_{k\in\Lambda}{\left(-\sum_{l\in\Lambda}{\left(L^0_l\right)^*{S^k_l}}\right)\,U_t\,da^k_0(t)}+\sum_{k,l\in\Lambda}{\left(S^k_l-\delta_{k,l}\,I_\Hcal\right)\,U_t\,da^k_l(t)}\,,
\end{multline}
where $H$ is selfadjoint and the operator $\Sbb=\sum_{i,j\in\Lambda}{S^i_j\otimes\proj{j}{i}}$ on $\Bcal(\Hcal\otimes\Kcal)$ is unitary. Then, after a change of noise of the form $f_i=We_i$, $i=1,\ldots, d$, the equation becomes
\begin{multline}\label{eqHPequnitayform2}
dU_t=-\left(i\,H+\frac{1}{2}\sum_{k\in\Lambda}{ \left(\tilde L^0_k\right)^*\,\tilde L^0_k}\right)\,U_t\,dt+\sum_{k\in\Lambda}{\tilde L^0_k\,U_t\,d\tilde a^0_k(t)}\hfill\\
\hfill +\sum_{k\in\Lambda}{\left(-\sum_{l\in\Lambda}{\left(\tilde L^0_l\right)^*\,{\tilde S^k_l}}\right)\,U_t\,d\tilde a^k_0(t)}+\sum_{k,l\in\Lambda}{(\tilde S^k_l-\delta_{k,l}\,I_\Hcal)\,U_t\,d\tilde a^k_l(t)}\,,
\end{multline}
where 
\begin{align*}
\tilde L^0&=W^*\, L^0\\
\tilde \Sbb&=W^* \, S\, W\,.
\end{align*}
\end{prop}

\subsection{The Noise Algebra}

In a previous article on one-step evolutions \cite{Bar2}, I. Bardet defined the Environment Algebra as the Von Neumann subalgebra of the environment generated by the unitary operator of a one-step evolution on the bipartite system $\Hcal\otimes\Kcal$. We recall here the basic definitions and the main result on the decomposition of the environment between a maximal commutative and a quantum part.

Let $\Sbb$ be a unitary operator on $\Hcal\otimes\Kcal$. For $f,g\in\Hcal$, we define:
\begin{equation}
\Sbb(f,g)=\Tr_{\proj{g}{f}}[\Sbb]\,,\qquad\Sbb^*(f,g)=\Tr_{\proj{g}{f}}[\Sbb^*]\,.
\label{eqU(f,g)}
\end{equation}
Those operators can be thought of as pictures of $\Sbb$ taken from $\Kcal$ but with different angles.
\begin{de}\label{deenvironment}
Let $\Sbb$ be a unitary operator on $\Hcal\otimes\Kcal$. We call \emph{Environment Algebra} the von Neumann algebra $\Acal(\Sbb)$ generated by the $\Sbb(f,g)$, that is
\begin{equation}
\Acal(\Sbb)=\left\lbrace \Sbb(f,g),\ \Sbb^*(f,g);\quad f,g\in\Hcal\right\rbrace''.
\label{eqdeenvironmentalg}
\end{equation}
\end{de}

The point with this definition is that it fits with the following characterization. 

\begin{prop}\label{propenvironment}
Let $\Sbb$ be a unitary operator on $\Hcal\otimes\Kcal$. Then $\Acal(\Sbb)$ is the smallest von Neumann subalgebra of $\Bcal(\Kcal)$ such that $\Sbb$ and $\Sbb^*$ belong to $\Bcal(\Hcal)\otimes\Acal(\Sbb)$, i.e. if $\Acal$ is another von Neumann algebra such that $\Sbb$ and $\Sbb^*$ belong to $\Bcal(\Hcal)\otimes\Acal$, then $\Acal(\Sbb)\subset\Acal$. Furthermore, its commutant is given by
\begin{equation}\label{eqpropenvironment}
\Acal(\Sbb)'=\left\lbrace Y\in\Bcal(\Kcal),\quad [I_\Hcal\otimes Y\,,\,\Sbb]=[I_\Hcal\otimes Y\,,\,\Sbb^*]=0\right\rbrace,
\end{equation}
where the notation $[\cdot,\cdot]$ stands for the commutant of two bounded operators.
\end{prop}
We now give the definition of the commutative parts of the environment and the resulting decomposition between a maximal commutative part and a quantum part.
\begin{de}\label{ClassicalQuantumpart}
Let $\Sbb$ be a unitary operator on $\Hcal\otimes\Kcal$ and let $\tilde{\Kcal}$ be a subspace of $\Kcal$. We say that $\tilde{\Kcal}$ is a \emph{Commutative Subspace of the Environment} if $\tilde{\Kcal}\ne\{0\}$ and:

\smallskip\noindent
i) $\Hcal\otimes\tilde{\Kcal}$ and $\Hcal\otimes\tilde{\Kcal}^\perp$ are stable by $\Sbb$,

\smallskip\noindent
ii) $\Acal(\tilde{\Sbb})$ is commutative, where $\tilde{\Sbb}$ is the restriction of $\Sbb$ to $\Hcal\otimes\tilde{\Kcal}$\,.
\end{de}
We then have the following Decomposition Theorem, which is proved in \cite{Bar2}.
\begin{theo}\label{theodecompenv}
The environment Hilbert space $\Kcal$ is the orthogonal direct sum of two subspaces $\Kcal_c$ and $\Kcal_q$, such that either $\Kcal_c=\{0\}$ or

\smallskip\noindent
i) $\Kcal_c$ is a commutative subspace of the environment.

\smallskip\noindent
ii) If $\tilde{\Kcal}$ is any commutative subspace of the environment, then $\tilde{\Kcal}\subset\Kcal_c$.

\smallskip\noindent
iii) The restriction of $\Sbb$ to $\Hcal\otimes\Kcal_q$ does not have any commutative subspace.
\end{theo}

\smallskip
We now come back to our continuous time scenario. In this situation, we define the \emph{Noise Algebra} as an algebra which encodes the structure of the noise in the unitary quantum Langevin equation.

\begin{de}\label{deenvironmentalgcont}
Let $(U_t)$ be the unitary-valued solution of a quantum Langevin Equation. The \emph{Noise Algebra} (at time $t$) is defined by
\begin{equation}\label{eqdeenvironmentalgcont}
\Acal_t(U)=\left\lbrace\ \Tr_{\proj{f}{g}}\left[U_s\right],\ \Tr_{\proj{f}{g}}\left[U_s^*\right];\quad f,g\in\Hcal,\quad 0<s\leq t\right\rbrace''\,.
\end{equation}
\end{de}

\noindent It is obvious that it is enough to consider only the vectors of an orthonormal basis of $\Hcal$ in this definition. Let $(g_i)_{i\in I}$ be such a basis. For simplicity we adopt the following notation: if $T$ is a bounded operator on $\Psi$, we write $T^{ij}$ for $\Tr_{\proj{f_i}{f_j}}\left[T\right]$, $i,j\in I$.

\section{The case of a Commutative Environment}\label{commEnv}

The aim of this section is to completely characterize those unitary quantum Langevin equations for which $\Acal_t(U)$ is commutative. We do that in subsection \ref{sect21}. The characterization is first algebraic, then interpreted in terms of classical noises. Finally, we apply this characterization to give the general form of the associated Lindblad generator in subsection \ref{sect22}. 

\subsection{Characterization of Commutative Noise Algebras}\label{sect21}

We shall need the following notations. If the matrix $\Sbb$, as a block-matrix on $\Kcal$, is diagonalizable in some orthonormal basis of $\Kcal$, we put $\Kcal_\Wbb$ to be the maximal subspace of $\Kcal$ such that $\Sbb$ acts as the identity operator on $\Hcal\otimes\Kcal_\Wbb$ and we put $\Kcal_\Pbb=\Kcal_\Wbb^\perp$. We shall see that whenever the environment is commutative, $\Kcal_\Wbb$ (resp. $\Kcal_\Pbb$) corresponds to the part where the noise is a Brownian process (resp. a Poisson process). We consider some  orthonormal basis $\{f_1,\ldots,f_{\tilde d}, f_{\tilde d+1},\ldots, f_d\}$ of $\Kcal$ adapted to the decomposition $\Kcal=\Kcal_\Wbb\oplus\Kcal_\Pbb$, where $m$ is the dimension of $\Kcal_\Wbb$. In this basis, the matrix $\Sbb$ can then be written as
\begin{equation}\label{decompunit}
\Sbb=\begin{pmatrix}
I_{\Hcal\otimes\Kcal_\Wbb} & 0 \\
0 & \Sbb_\Prob
\end{pmatrix},\quad \Sbb_\Prob=\begin{pmatrix}
S_1 & 0 & \cdots & 0 \\
0 & S_2 & \ddots & \vdots\\
\vdots & \ddots & \ddots & 0 \\
0 & \cdots & 0 & S_{d-m}
\end{pmatrix}\,.
\end{equation}
 We shall denote by $\Lambda_\Wbb$ the set of indices $\{1, \ldots, \tilde d\}$ and by $\Lambda_\Pbb$ the other one. 
 
\smallskip
We can now state the first main result of this article, which completely characterizes the commutativity of the algebra $\Acal_t(U)$ in terms of algebraic properties of the coefficients $L^i_j$. 

\begin{theo}\label{theocomnoisealgebra}
Consider a unitary-valued quantum Langevin equation of the form
\begin{multline}\label{eqHPequnitaryform}
dU_t=-\left(i\,H+\frac{1}{2}\sum_{k\in\Lambda}{(L^0_k)^*\,L^0_k}\right)\,U_t\,dt+\sum_{k\in\Lambda}{L^0_k\,U_t\,da^0_k(t)}\hfill\\
\hfill +\sum_{k\in\Lambda}{\left(-\sum_{l\in\Lambda}{(L^0_l)^*\,{S^k_l}}\right)\,U_t\,da^k_0(t)}+\sum_{k,l\in\Lambda}{\left(S^k_l-\delta_{k,l}\,I_\Hcal\right)\,U_t\,da^k_l(t)}\,,
\end{multline}
where $H$ is selfadjoint and the operator $\Sbb=\sum_{i,j\in\Lambda}{S^i_j\otimes\proj{j}{i}}$ on $\Bcal(\Hcal\otimes\Kcal)$ is unitary. Then the following assertions are equivalent. 

\smallskip\noindent
1) The algebra $\Acal_t(U)$ is commutative for all $t>0$.

\smallskip\noindent
2) The algebra $\Acal_t(U)$ is commutative for some $t>0$.

\smallskip\noindent 
3) The matrix $\Sbb$, as a block-matrix on $\Kcal$, is diagonalizable in some orthonormal basis of $\Kcal$ (which is equivalent to $\Acal(\Sbb)$ commutative) and, considering the coefficients $L^i_j$ after the appropriate change of noise, we have that

\smallskip
i) there exists a symmetric unitary operator $W$ on $\Kcal_\Wbb$ such that
\begin{equation}\label{LWL}
\begin{pmatrix}
\left(L^0_1\right)^*\\\vdots\\\left(L^0_m\right)^*
\end{pmatrix}
=W\begin{pmatrix}
L^0_1\\\vdots\\L^0_m
\end{pmatrix}, 
\end{equation}

\smallskip
ii) for all $i\in\Lambda_\Prob$, there exists $\lambda_i\in\C$ such that
\begin{equation}\label{eqtheocomalg1}
L^0_i=\lambda_i(S_i-I_\Hcal)\,.
\end{equation}

\smallskip\noindent
4) There exists a change of noise such that Equation \eqref{eqHPequnitaryform} is of the form
\begin{equation}\label{Ubruits}
dU_t=A_0\,U_t\, dt+\sum_{i\in\Lambda_\Wbb} A_i\, U_t\, dW^i_t+\sum_{i\in\Lambda_\Pbb} B_i\, U_t\, dX^i_t\,,
\end{equation}
where $A_0$, $\{A_i,\,i\in\Lambda_\Wbb\}$ and $\{B_i,\,i\in \Lambda_\Pbb\}$ are bounded operators on $\Hcal$, where $W^i$, $i\in\Lambda_\Wbb$, are (multiplication operators by) standard Brownian motions, where $X^i$, $i\in\Lambda_\Pbb$, are (multiplication operators by) compensated Poisson processes and such that all the processes $W^i$ and $X^j$ are independent. 
\end{theo}

\begin{proof}
Obviously 1) implies 2). 

\smallskip\noindent
\underline{Proof of 2) implies 3) :}

 Let us write Equation \eqref{eqHPequnitaryform} as 
$$
dU_t=\sum_{i,j\in\Lambda\cup\{0\}} L^i_j\, U_t\, da^i_j(t)\,,
$$
for short.  

We consider a fixed orthonormal basis $(g_i)_{i\in I}$ of $\Hcal$ and the corresponding notation $T^{ij}$ for bounded operators on $\Psi$. We first exploit the relation $[U^{kl}_s\,,\,U^{mn}_s]=0$ for all $s<t$, all $k,l,m,n\in I$. Differentiating this equality and using the It\^o rule we get
\begin{align}
0&=
\sum_{i,j\in\Lambda\cup\{0\}} \left[  \left(L^i_jU_s\right)^{kl}U^{mn}_s   +U^{kl}_s\left(L^i_jU_s\right)^{mn}+\sum_{i'\in\Lambda}  \left(L^{i'}_jU_s\right)^{kl}\left(L^i_{i'}U_s\right)^{mn}\right]\, da^i_j(s)\nonumber\\
&\ \ \ -\sum_{i,j\in\Lambda\cup\{0\}} \left[  \left(L^i_jU_s\right)^{mn}U^{kl}_s   +U^{mn}_s\left(L^i_jU_s\right)^{kl}+\sum_{i'\in\Lambda}  \left(L^{i'}_jU_s\right)^{mn}\left(L^i_{i'}U_s\right)^{kl}\right]\, da^i_j(s)\,.\label{dUdU}
\end{align}
Identifying each of the coefficients of $da^i_j(s)$ to 0 and taking the limit $s\rightarrow 0$, we get, for all $k,l,m,n\in I$, all $i,j\in\Lambda\cup\{0\}$, 
$$
\left(L^i_j\right)^{kl} +\left(L^i_j\right)^{mn}+\sum_{i'\in\Lambda}  \left(L^{i'}_j\right)^{kl}\left(L^i_{i'}\right)^{mn}- \left(L^i_j\right)^{mn}  -\left(L^i_j\right)^{kl}-\sum_{i'\in\Lambda}  \left(L^{i'}_j\right)^{mn}\left(L^i_{i'}\right)^{kl}=0\,,
$$
that is,
\begin{equation}\label{LL}
\sum_{i'\in\Lambda}  \left(L^{i'}_j\right)^{kl}\left(L^i_{i'}\right)^{mn}-\left(L^{i'}_j\right)^{mn}\left(L^i_{i'}\right)^{kl}=0\,.
\end{equation}
Consider the block-matrix $\Lbb$ given by $\Lbb_{ij}=L^j_i$, $i,j\in\Lambda$. We denote by ${\Lbb}^{kl}$ the matrix with coefficients $(\Lbb^{kl})_{ij}=(\Lbb_{ij})^{kl}=(L^j_i)^{kl}$. With these notations we have 
$$
\left(\Lbb^{kl}\Lbb^{mn}\right)_{ij}=\sum_{i'\in\Lambda}\left(\Lbb^{kl}\right)_{ii'}\left(\Lbb^{mn}\right)_{i'j}=\sum_{i'\in\Lambda}\left(L^{i'}_i\right)^{kl}\left(L^j_{i'}\right)^{mn}\,.
$$
This way, Equation \eqref{LL} means 
\begin{equation}\label{LL2}
\left[\Lbb^{kl},\Lbb^{mn}\right]=0\,.
\end{equation}
 
We now exploit the relation $[(U^*_s)^{kl}\,,\,U^{mn}_s]=0$ for all $s<t$, all $k,l,m,n\in I$. Differentiating this equality and using the It\^o rule we get
\begin{align*}
0&=
\sum_{i,j\in\Lambda\cup\{0\}} \left[  \left(U_s^*(L^j_i)^*\right)^{kl}U^{mn}_s   +(U^*_s)^{kl}\left(L^i_jU_s\right)^{mn}+\sum_{i'\in\Lambda}  \left(U_s^*(L^{j}_{i'})^*\right)^{kl}\left(L^i_{i'}U_s\right)^{mn}\right]\, da^i_j(s)\\
&\ \ \ -\sum_{i,j\in\Lambda\cup\{0\}} \left[  \left(L^i_jU_s\right)^{mn}(U^*_s)^{kl}   +U^{mn}_s\left(U_s^*(L^j_i)^*\right)^{kl}+\sum_{i'\in\Lambda}  \left(L^{i'}_jU_s\right)^{mn}\left(U_s^*(L^{i'}_{i})^*\right)^{kl}\right]\, da^i_j(s)\,.
\end{align*}
Identifying each of the coefficients of $da^i_j(s)$ to 0 and taking the limit $s\rightarrow 0$, we get, for all $k,l,m,n\in I$, all $i,j\in\Lambda\cup\{0\}$, 
$$
\left((L^j_i)^*\right)^{kl} +\left(L^i_j\right)^{mn}+\sum_{i'\in\Lambda}  \left((L^{j}_{i'})^*\right)^{kl}\left(L^i_{i'}\right)^{mn}- \left(L^i_j\right)^{mn}  -\left((L^j_i)^*\right)^{kl}-\sum_{i'\in\Lambda}  \left(L^{i'}_j\right)^{mn}\left((L^{i'}_{i})^*\right)^{kl}=0\,,
$$
that is,
\begin{equation}\label{LsL}
\sum_{i'\in\Lambda}  \left((L^{j}_{i'})^*\right)^{kl}\left(L^i_{i'}\right)^{mn}-\left(L^{i'}_j\right)^{mn}\left((L^{i'}_{i})^*\right)^{kl}=0\,.
\end{equation}
The block-matrix $\Lbb$ defined above satisfies $(\Lbb^*)_{ij}=(\Lbb_{ji})^*=(L^i_j)^*$, so that 
$$
\left((\Lbb^*)^{kl}\right)_{ij}=\left((\Lbb^*)_{ij}\right)^{kl}=\left((L^i_j)^*\right)^{kl}\,.
$$ 
This way, Equation \eqref{LsL} means 
\begin{equation}\label{LsL2}
\left[(\Lbb^*)^{kl},\Lbb^{mn}\right]=0\,.
\end{equation}

With Equations \eqref{LL2} and \eqref{LsL2} we have proved that all the matrices $\Lbb^{kl}$ and $(\Lbb^*)^{mn}$ commute. As $\Sbb=\Lbb+I_{\Hcal\otimes\Kcal}$, we get that all the matrices $\Sbb^{kl}$ and $(\Sbb^*)^{mn}$ commute too, so that the algebra $\Acal(\Sbb)$ is commutative.

As a consequence the block-matrix $\Sbb$ can be block-diagonalized. As announced before the theorem, $\Kcal$ can be decomposed as $\Kcal=\Kcal_\Wbb\oplus\Kcal_\Pbb$ and the matrix $\Sbb$ can then be written as in Equation \eqref{decompunit}: 
\begin{equation*}
\Sbb=\begin{pmatrix}
I_{\Hcal\otimes\Kcal_\Wbb} & 0 \\
0 & \Sbb_\Prob
\end{pmatrix}
\end{equation*}
We now make a change of noise adapted to the decomposition of $\Kcal$ as the direct sum of $\Kcal_\Wbb$ and $\Kcal_\Pbb$. In particular, in the new noises, we have
\begin{equation*}
L^i_j=0 \text{ for all }i,j\in\Lambda,\ i\ne j\,.
\end{equation*}
Following our notations above, we put $S_i=S^i_i$ for all $i\in\Lambda$. Note that the coefficients $S_i$  have to be  unitary operators on $\Hcal$, for $\Sbb$ to be unitary on $\Hcal\otimes\Kcal$.  

\smallskip
With these reductions, Equation \eqref{LL} becomes, when $i=0$ and $j\ne0$
\begin{equation}\label{LL0}
\left(S_j-I\right)^{kl}\left(L^0_j\right)^{mn}=\left(S_j-I\right)^{mn}\left(L^0_j\right)^{kl}\,.
\end{equation}
On the other hand, Equation \eqref{LsL} reduces, when $i=j=0$, to
\begin{equation}\label{LsL0}
\sum_{i\in\Lambda} \left((L^0_i)^*\right)^{kl}\left(L^0_i\right)^{mn}=\sum_{i\in\Lambda}\left(L^i_0\right)^{mn}\left((L^i_0)^*\right)^{kl}\,.
\end{equation}

Let us consider some index $i\in \Lambda_\Pbb$, that is, for which $S_i\not = I$. In particular there exist $k,l\in I$ such that $\left(S_i-I_\Hcal\right)^{kl}\ne0$. Equation \eqref{LL0} then gives for all $m,n\in I$, 
\begin{equation*}
\left(L^0_i\right)^{mn}=\frac{\left(L^0_i\right)^{kl}}{\left(S_i-I\right)^{kl}}\left(S_i-I\right)^{mn}\,.
\end{equation*}
Defining $\lambda_i=\left(L^0_i\right)^{kl}/\left(S_i-I\right)^{kl}$, this gives $L^0_i=\lambda_i\, (S_i-I)$. This proves the property \eqref{eqtheocomalg1}.

\smallskip
We now come back to Equation \eqref{LsL0}, separating the indices in $\Lambda_\Pbb$ and those in $\Lambda_\Wbb$. Using the fact that $L^i_0=-(L^0_i)^*$ when $i$ belongs to $\Lambda_\Wbb$ and the relation $L^0_i= {\lambda_i}(S_i-I)$ when $i$ belongs to $\Lambda_\Pbb$,  we get from Equation \eqref{LsL0}
\begin{multline*}
\sum_{i\in\Lambda_\Wbb} \left((L^0_i)^*\right)^{kl}\left(L^0_i\right)^{mn}+\sum_{i\in\Lambda_\Pbb} |\lambda_i|^2 \left((S_i-I)^*\right)^{kl}\left(S_i-I\right)^{mn}=\hfill\\
\hfill=\sum_{i\in\Lambda_\Wbb}\left((L^0_i)^*\right)^{mn}\left(L^0_i\right)^{kl}+\sum_{i\in\Lambda_\Pbb} |\lambda_i|^2 \left((S_i-I)^*\right)^{kl}\left(S_i-I\right)^{mn}\,.
\end{multline*}
This reduces to
$$
\sum_{i\in\Lambda_\Wbb} \left((L^0_i)^*\right)^{kl}\left(L^0_i\right)^{mn}=\sum_{i\in\Lambda_\Wbb}\left((L^0_i)^*\right)^{mn}\left(L^0_i\right)^{kl}\,, 
$$
or else
\begin{equation}\label{LsL0p}
\sum_{i\in\Lambda_\Wbb} \overline{\left(L^0_i\right)^{lk}}\left(L^0_i\right)^{mn}=\sum_{i\in\Lambda_\Wbb}\overline{\left(L^0_i\right)^{nm}}\left(L^0_i\right)^{kl}\,.
\end{equation}
Put $u(k,l)=((L^0_i)^{kl})_{i\in\Lambda_\Wbb}\in\C^{\tilde d}$ and $v(k,l)=(\overline{(L^0_i)^{lk}})_{i\in\Lambda_\Wbb}\in\C^{\tilde d}$, for all $k,l\in I$. Equation \eqref{LsL0p} then becomes
$$
\langle u(l,k)\,,\,u(m,n)\rangle=\langle v(l,k)\,,\, v(m,n)\rangle\,,
$$
for all $k,l,m,n\in I$. 

We claim that this implies that there exists a unitary operator $W$ on $\C^{\tilde d}$ such that $W\, u(k,l)=v(k,l)$ for all $k,l\in I$. Indeed, we can assume that the family $\{u(k,l)\,;\ k,l\in I\}$ has maximal rank in $\C^{\tilde d}$, otherwise we complete it. The family $\{v(k,l)\,;\ k,l\in I\}$ has same rank, so we complete it in the same way. We denote by $\Lcal^2(\Hcal)$ the class of Hilbert-Schmidt operators on $\Hcal$. We recall that it is a Hilbert space when associated to the inner product given by the trace. Consider the operators $U$ and $V$ from $\Lcal^2(\Hcal)$ to $\C^{\tilde d}$ defined by:
$$
\begin{array}{cccc}
U: & \Lcal^2(\Hcal) & \to & \C^{\tilde d}\\
& \proj{k}{l} & \mapsto & u(k,l)
\end{array}
$$
(and the same for $V$ with $u(k,l)$ replaced with $v(k,l)$). By hypothesis we have $V^*V=U^*U$. Consider the polar decomposition of $U$ and $V$:
$$
U=M\sqrt{U^*U}\qquad\mbox{and}\qquad V=N\sqrt{V^*V}\,,
$$
where $M,N\,:\,\Lcal^2(\Hcal) \rightarrow C^{\tilde d}$ are partial isometries, where $\Ker M=\Ker U$ and $\Ker N=\Ker V$, where $\Ima M=\Ima U=\C^{\tilde d}=\Ima V=\Ima N$. Put $W_1$ to be a unitary operator on $\C^{\tilde d}$ which agrees with $M$ on $(\Ker M )^\perp$ and $W_2$ to be a unitary operator on $\C^{\tilde d}$ which agrees with $N$ on $(\Ker N )^\perp$. Putting $W=W_2W_1^*$ it is easy to check that $V=WU$. This proves the claim. 

Now, let us prove that $W$ has also to be symmetric, that is $W^t=W$. We have proved the relation $V=WU$, that is for all $k,l\in I$
$$
\begin{pmatrix}\left((L^0_1)^*\right)^{kl}\\\vdots \\\left((L^0_m)^*\right)^{kl}\end{pmatrix}=W\, \begin{pmatrix}\left(L^0_1\right)^{kl}\\\vdots \\\left(L^0_m\right)^{kl}\end{pmatrix}\,.
$$
Hence  we have
$$
\begin{pmatrix}(L^0_1)^*\\\vdots \\(L^0_m)^*\end{pmatrix}=W\, \begin{pmatrix}L^0_1\\\vdots \\ L^0_m\end{pmatrix}\,,
$$
so that
$$
\begin{pmatrix}L^0_1&\ldots &L^0_m\end{pmatrix}=\begin{pmatrix}(L^0_1)^*&\ldots &(L^0_m)^*\end{pmatrix}\, W^*
$$
and finally
$$
\begin{pmatrix}L^0_1\\\vdots \\L^0_m\end{pmatrix}=\overline{W}\, \begin{pmatrix}(L^0_1)^*\\\vdots \\ (L^0_m)^*\end{pmatrix}\,.
$$
We have proved that $\overline{W}=W^*$, hence $W^t=W$. We have proved the property \eqref{LWL}. We have proved that 1) implies 3). 

\smallskip\noindent
\underline{Proof of 3) implies 4) :}
 
If the coefficients $L^i_j$ satisfy all the properties described in 3), then, after the adequate change of noise, Equation \eqref{eqHPequnitaryform} reduces to
\begin{multline}\label{34}
dU_t=K_0\,U_t\,dt+\sum_{k\in\Lambda_\Wbb}{L^0_k\,U_t\,da^0_k(t)}+\sum_{k\in\Lambda_\Wbb} L^k_0\, U_t\, da^k_0(t)+\hfill\\
\hfill + \sum_{k\in\Lambda_\Pbb} \lambda_k \left(S_k-I\right)\, U_t\, da^0_k(t)+\sum_{k\in\Lambda_\Pbb} \overline{\lambda_k }\left(S_k-I\right)\, U_t\, da^k_0(t)+\sum_{k\in\Lambda_{\Pbb}} \left(S_k-I\right)\, U_t\, da^k_k(t)\,,
\end{multline}
where we do not need to detail the operator $K_0$ anymore here. We first focus on the Wiener part. Write $L_0$ the column of the coefficients $L^0_i$, $i\in\Lambda_\Wbb$ and $L^0$ the row of the coefficients $L^i_0$, $i\in\Lambda_\Wbb$. Recall that $L_0=-(L^0)^*$, but we also have as a condition in 3) that $((L^0)^*)^t=W(L^0)^*$, or else $(L^0)^*=(L^0)^t\, W^t$. 

The matrix $W$ is a symmetric unitary operator. By the Takagi Factorization theorem (see \cite{HJ12} for example), every symmetric complex square matrix can be decomposed as $V^t\, D\, V$, where $V$ is unitary and $D$ is diagonal with real positive entries. Thus $W$ admits such a decomposition. But the fact that $W$ is unitary gives
$$
I=W^*W=V^*\,D\, \overline V\, V^t \, D\, V=V^*\, D^2\, V\,.
$$
In particular $D^2=I$ and thus $D=I$. We have proved that $W$ is the form $V^t\, V$ for some unitary $V$. Actually, we apply this decomposition to $-W$ instead and we write $-W=V^t\, V$ for some unitary matrix $V$. 

As we said above, we have $L_0=-(L^0)^t\, W^t$, which gives
$$
L_0=(L^0)^t\,V^t \, V=(V\, L^0)^t\, V\,.
$$
On the other hand we have $L^0=V^*\,(V\,L^0)$. Hence if we put $K=V\, L^0$, the part 
$$
\sum_{k\in\Lambda_\Wbb}{L^0_k\,U_t\,da^0_k(t)}+\sum_{k\in\Lambda_\Wbb} L^k_0\, U_t\, da^k_0(t)
$$
of Equation \eqref{34} can now be written as
$$
\sum_{k\in\Lambda_\Wbb}{(V^*K)_k\,U_t\,da^0_k(t)}+\sum_{k\in\Lambda_\Wbb} ((K^t)V)_k\, U_t\, da^k_0(t)\,.
$$
Hence, by Proposition \ref{change}, if we apply the change of noise associated to $V$ to the orthonormal basis of $\Kcal_\Wbb$, we obtain a term of the form
$$
\sum_{k\in\Lambda_\Wbb}{K_k\,U_t\,da^0_k(t)}+\sum_{k\in\Lambda_\Wbb} K_k\, U_t\, da^k_0(t)=\sum_{k\in\Lambda_\Wbb}{K_k\,U_t\,dW^k_t}.
$$
Let us now concentrate on the part indexed by $\Lambda_\Pbb$ in Equation \eqref{34}. If we decompose each $\lambda_k$ into $\rho_k\, e^{i\theta_k}$, for those $\lambda_k\not=0$, then this part of \eqref{34} can be written as
$$
\sum_{k\in\Lambda_\Pbb} \rho_k \left(S_k-I\right)\, U_t\, \left(e^{i\theta_k}\, da^0_k(t)+e^{-i\theta_k}\, da^k_0(t)+\frac{1}{\rho_k}\,  da^k_k(t)\right)\,.
$$
After a change of noise $f_k\mapsto e^{i\theta_k}\, f_k$, the expression above reduces to
$$
\sum_{k\in\Lambda_\Pbb} \rho_k \left(S_k-I\right)\, U_t\, \left(da^0_k(t)+ da^k_0(t)+\frac{1}{\rho_k}\,  da^k_k(t)\right)\,,
$$
which is exactly of the form announced in 4), that is,
$$
 \sum_{k\in\Lambda_\Pbb} B_k\, U_t\, dX^k_t\,.
$$
For those $\lambda_k=0$, as we discussed in Subsection \ref{notations}, the equation gives rise to a term of the form
$$
\left(S_k-I\right)\, U_t\,  da^k_k(t)
$$
which is of no contribution in the probabilistic interpretation of the equation, nor the effect of the evolution on the small system $\Hcal$. 

\smallskip\noindent
\underline{Proof of 4) implies 1) :}
We start with a quantum Langevin equation of the form of Equation \eqref{Ubruits}
\begin{equation*}
dU_t=A_0\,U_t\, dt+\sum_{i\in\Lambda_\Wbb} A_i\, U_t\, dW^i_t+\sum_{i\in\Lambda_\Pbb} B_i\, U_t\, dX^i_t\,,
\end{equation*}
as stated in 4). 

Consider the von Neumann algebra $\Ncal_t$ generated by the operators $\{W^i_s\,,\, X^i_s\,;\ i\in\Lambda\,,\, s\leq t\}$.
% that is,\todo{A verifier}
%$$
%\Ncal_t=\left\{\bigotimes_{i\in\Lambda_\Wbb}W^i_s\otimes\bigotimes_{i\in\Lambda_\Pbb}X^i_s;\quad s\leq t\right\}''.
%$$
It is clear that it is a commutative von Neumann algebra, for they are all self-adjoint and pairwise commuting operators. We claim that $\Acal_t(U)\subset\Ncal_t$, let us prove this fact. 

Indeed, with the conditions on the coefficients of Equation \eqref{Ubruits}, it is rather standard to prove that the solution $(U_t)$ can be obtained as the strong limit of a Picard iteration:
\begin{align*}
U_s^0&=I\,,\ \mbox{for all } s\leq t\,,\\
U_t^{n+1}&=I+\int_0^t A_0\,U_s^n\, ds+\sum_{i\in\Lambda_\Wbb} \int_0^t A_i\, U_s^n\, dW^i_s+\\
&\ \ \ +\sum_{i\in\Lambda_\Pbb}\int_0^t B_i\, U_s^n\, dX^i_s\,.
\end{align*}
Clearly all the $U_s^0$, $s\leq t$, belong to $\Bcal(\Hcal)\otimes\Ncal_t$. By induction, if all the $U_s^n$, $s\leq t$, belong to $\Bcal(\Hcal)\otimes\Ncal_t$, then so do all the $U_s^{n+1}$, $s\leq t$, for in the equation above $U_s^{n+1}$ is obtained as strong limit of Riemann sums of elements of $\Bcal(\Hcal)\otimes\Ncal_t$. Hence, passing to the strong limit, all the $U_s$, $s\leq t$, belong to $\Bcal(\Hcal)\otimes\Ncal_t$. The same argument works also for the $U_s^*$. As $\Acal_t(U)$ is the smallest von Neumann algebra such that $U_s\in\Bcal(\Hcal)\otimes\Acal_t(U)$ for all $s\leq t$ and the same for the $U_s^*$, we get the announced inclusion: $\Acal_t(U)\subset \Ncal_t$. 

As a consequence, $\Acal_t(U)$ is commutative, and this holds for all $t\in \R^+$. This proves 1). 

\smallskip
Strictly speaking, Equation \eqref{Ubruits} may contain additional terms with $da^k_k(t)$ alone as a driving noise. These terms corresponding to the cases $\lambda_k=0$ in 3). But these terms change nothing on the proof of "4) implies 1)", as we can add them to the commutative algebra $\Ncal_t$ and carry on with exactly the same proof.

\smallskip 
The theorem is proved. 
\end{proof}

\smallskip
We thought it could be of interest to make explicit the same theorem as above, but in the case $d=1$, for in this case it takes a particularly simple form. 

\begin{theo}\label{theocasd=1}
Consider the unitary solution $(U_t)$ of a quantum Langevin equation:
\begin{equation*}
dU_t=\left(iH+\frac{1}{2}L^*L\right)\, U_t\,dt+L\, U_t\, da^0_1(t)-L^*\, S\, U_tda^1_0(t)+(S-I_\Hcal)\,U_t\,da^1_1(t),
\end{equation*}
with $H,L,S\in\Bcal(\Hcal)$, $H=H^*$ and $S$ a unitary operator. Then the following assertions are equivalent:

\smallskip\noindent
1) $\Acal_t(U)$ is commutative for some $t>0$.

\smallskip\noindent
2) $\Acal_t(U)$ is commutative for all $t>0$.

\smallskip\noindent 
3) One of the following two conditions holds:

\noindent -- either $S=I_\Hcal$ and there exists $\theta\in\R$ such that $L^*=e^{i\theta}L$;

\noindent -- or there exists a complex number $\lambda$ such that $L=\lambda(S-I_\Hcal)$.

\smallskip\noindent
4) After the appropriate change of noise Equation \eqref{eqHPintro} takes the form of either Equation (\ref{eqHPintrobruitbrownian}) or Equation (\ref{eqHPintrobruitpoisson}) with a Poisson process of intensity $|\lambda|$ depending whether the first or the second case holds in 3).
\end{theo}

\subsection{Lindblad Generators}\label{sect22}

One of the main motivation in constructing the unitary solution of quantum Langevin Equations is that their solutions give cocycle unitary dilations of Quantum Markov Semigroups (QMS). Indeed, let $U_\cdot$ be the unique unitary solution of Equation (\ref{eqHPequnitaryform}). Using the It\^o table for the quantum noises one shows that if we put
\begin{equation}
\Pcal_t(X)= \sca{\Omega}{U_t^* \left(X\otimes I_\Phi\right) U_t\, \Omega}\qquad X\in\Bcal(\Hcal)\,,
\label{eqQMSdil}
\end{equation}
for all $t\in\R^+$, then this defines a norm-continuous Quantum Markov Semigroup on $\Bcal(\Hcal)$. Moreover, its Lindblad generator $\Lcal(\cdot)$ is given in the Heisenberg picture by
\begin{equation}
\Lcal(X)=-i[H,X]+ \frac{1}{2}\sum_{k\in\Lambda}\left(-L_k^*L_kX-XL_k^*L_k+2L_k^*XL_k\right).
\label{eqLimbladgene}
\end{equation}

We see that the unitary operator $\Sbb$ does not play any role in this generator. For this reason, it is called the \emph{gauge} of the quantum Langevin Equation. 

Any generator of a norm-continuous QMS on $\Bcal(\Hcal)$ can be written under the form (\ref{eqLimbladgene}), so that the QMS admits a cocycle unitary dilation $U_\cdot$ solution of a quantum Langevin equation.

We now illustrate Theorem \ref{theocomnoisealgebra} with two applications to QMS: essentially commutative dilation and detailed balance condition.

A result of Kummerer and Maassen \cite{K-M4} characterizes the particular structure of those Lindblad generators for which the QMS admits an essentially commutative dilation. In their work, the term ``dilation" is more general than quantum Langevin equations. However, within our framework, we are able to obtain their result, as is stated in the following Theorem.

\begin{theo}\label{theoLindcom}
Let $\Pcal_\cdot$ be a QMS on $\Bcal(\Hcal)$. Then the following are equivalent.

\smallskip\noindent
1) The semigroup $\Pcal_\cdot$ admits a dilation $U_\cdot$, solution of a unitary quantum Langevin equation, such that $\Acal_{t}(U_\cdot)$ is commutative, for all $t>0$. 

\smallskip\noindent
2) There exist

\noindent
--  selfadjoint operators $H,L_1,...,L_m$ on $\Hcal$,

\noindent
-- unitary operators $S_1,...,S_n$ on $\Hcal$,

\noindent
-- positive real numbers $\lambda_1,...,\lambda_n$,

\noindent
such that the Lindblad generator $\Lcal$ of $\Pcal_\cdot$ is given by:
\begin{equation}\label{eqtheoLindcom}
\Lcal(X)=-i[H,X]+ \frac{1}{2}\sum_{k=1}^m\left(2L_kXL_k-L_k^2X-XL_k^2\right)+\sum_{k=1}^n{\lambda_i\left(S_k^*XS_k-X\right)}\,.
\end{equation}
\end{theo}

\begin{proof}
The proof is an immediate consequence of Theorem \ref{theocomnoisealgebra}. Indeed, if $\Pcal_\cdot$ admits a  dilation $U_\cdot$ such that $\Acal_{t}(U)$ is commutative for all $t\geq0$, then by Theorem \ref{theocomnoisealgebra} a change of noise leads to a quantum Langevin Equation for $U_\cdot$ of the form of Equation \eqref{Ubruits} and the result follows.

Conversely, if the Lindblad generator is of the form 2), the quantum Langevin equation with the corresponding coefficients has its algebra $\Acal_t(U)$ commutative for all $t$. 
\end{proof}

\medskip
The other link with QMS we want to mention concerns the detailed balance condition as defined in \cite{F-U1}. When the invariant state is the normalized trace, this condition summarizes into
\begin{equation*}
\Tr\left[\Lcal(X)Y\right]-\Tr\left[X\Lcal(Y)\right]=\Tr\left[XYK-YXK\right]\text{ for all }X,Y\in\Bcal(\Hcal),
\end{equation*}
where $K\in\Bcal(\Hcal)$ is a selfadjoint operator. Fagnola and Umanita proved in \cite{F-U1} that detailed balance condition with respect to the normalized trace holds if and only if there exists a representation of the Limblad generator such that $L_k=L_k^*$ for all $k\in\Lambda$. As a direct consequence of
Theorem \ref{theocomnoisealgebra} we obtain:

\begin{theo}
A QMS satisfies the detailed balance condition with respect to the normalized trace on $\Bcal(\Hcal)$ if and only if its admits a unitary dilation  which is the solution of a classical Langevin Equation with Brownian noises only.
\end{theo}

\section{Classical and Quantum parts of the Environment}

We are now back to general quantum Langevin equation and we shall prove that they can always be splitted into a maximal commutative part and a purely quantum one.

\subsection{The Decomposition Theorem}

In this section we study the Noise Algebra $\Acal_t(U)$ in the general case. Note that if $\Kcal=\Kcal_1\oplus\Kcal_2$, then by the exponential property of the symmetric Fock space one has

\begin{equation*}
\Phi=\Gamma_s\left(L^2(\R^+,\Kcal_1)\right)\otimes\Gamma_s\left(L^2(\R^+,\Kcal_2)\right).
\end{equation*}

\noindent For short, if $\tilde{\Kcal}$ is a subspace of $\Kcal$, we write $\Phi(\tilde{\Kcal})=\Gamma_s\left(L^2(\R^+,·\tilde{\Kcal})\right)$. Suppose that $U_\cdot$ is the unitary solution of Equation (\ref{eqHPequnitaryform}) and that both $\Hcal\otimes\Kcal_1$ and $\Hcal\otimes\Kcal_2$ are stable by $\Sbb$. Then, as already mentioned before, $\Lambda=\Lambda_1\cup\Lambda_2$ accordingly to the decomposition of $\Sbb$ and Equation (\ref{eqHPequnitaryform}) can be written as:

\begin{equation}\label{eqUdecomp}
d U_t=-\left(iH+\frac{1}{2}\sum_{i\in\Lambda}{(L^0_i)^*L^0_i}\right)U_t\,dt + \sum_{\substack{i,j\in\Lambda_1\cup\{0\}\\ i+j\ne0}}L^i_jU_tda_j^i(t) + \sum_{\substack{i,j\in\Lambda_2\cup\{0\}\\ i+j\ne0}}L^i_jU_tda_j^i(t).
\end{equation}

\begin{de}
Let $\Kcal_1$ be a subspace of $\Kcal$ and write $\Kcal_2=\Kcal_1^\perp$. We say that \emph{$\Phi(\Kcal_1)$ is a Commutative Subsystem of the Environment} if $\Kcal_1\ne\{0\}$ and:
\begin{itemize}
\item both $\Hcal\otimes\Kcal_1$ and $\Hcal\otimes\Kcal_2$ are stable by $\Sbb$. Consequently, up to a change of noise, Equation \eqref{eqHPequnitaryform} takes the form of Equation \eqref{eqUdecomp}.
\item $\Acal_t(U^1)$ is commutative for some $t>0$, where $U^1_\cdot$ is the unique unitary solution of the quantum Langevin equation:
\begin{equation*}
d U^1_t=-\frac{1}{2}\sum_{i\in\Lambda_1}\left((L^0_i)^*L^0_i\right)U^1_t\,dt + \sum_{\substack{i,j\in\Lambda_1\cup\{0\}\\ i+j\ne0}}L^i_jU^1_tda_j^i(t)
\end{equation*}
(i.e. we consider the quantum Langevin equation with only the coefficients indexed by $\lambda_1$).
\end{itemize}
\label{declassicalpart}
\end{de}

Consequently, using the notation of the previous definition, if $\Phi(\Kcal_1)$ is a commutative Subsystem of the Environment then Theorem \ref{theocomnoisealgebra} can be applied to $U^1_\cdot$, so that it obeys a classical stochastic differential equation driven by independent Poisson processes and Brownian Processes. This in turn implies conditions on the coefficients in Equation \eqref{eqUdecomp}.

\begin{theo}[Decomposition Theorem]
Suppose that $U_\cdot$ is the unique unitary solution of Equation (\ref{eqHPequnitaryform}). Then $\Kcal$ is the orthogonal direct sum of two subspaces $\Kcal_c$ and $\Kcal_q$ such that either $\Kcal_c=\{0\}$ or:
\begin{itemize}
\item $\Phi(\Kcal_c)$ is a Commutative Subsystem of the Environment.
\item If $\tilde{\Kcal}$ is a subspace of $\Kcal$ such that $\Phi(\tilde{\Kcal})$ is a Commutative Subsystem of the Environment, then $\tilde{\Kcal}$ is a subspace of $\Kcal_c$.
\item $U^q_\cdot$ does not have any Commutative Subsystem, where $U^q_\cdot$ is the unique unitary solution of the quantum Langevin equation:
\begin{equation*}
d U^q_t=-\frac{1}{2}\sum_{i\in\Lambda_q}\left((L^0_i)^*L^0_i\right)U^q_t\,dt + \sum_{\substack{i,j\in\Lambda_q\cup\{0\}\\ i+j\ne0}}L^i_jU^q_tda_j^i(t).
\end{equation*}
\end{itemize}
\label{theodecompcont}
\end{theo}

\begin{proof}
The first step of the proof is to identify the subspace $\Kcal_c$. To do that, let $\Pcal_c$ be the set of orthogonal projections $P\in\Acal(\Sbb)'$ such that $P\in\Pcal_c$ iff $\Phi(P\Kcal)$ is a Commutative Subsystem of the Environment. We claim that $\Pcal_c$ has a maximal element. Indeed, as $\Kcal$ is finite dimensional, for any totally order set $\Pcal\subset\Pcal_c$ there exists a projection $P_{\max}\in\Pcal$ such that $P_{\max}\Kcal$ has the highest dimension for this set. Consequently $P\leq P_{\max}$ for all $P\in\Pcal$. Thus $\Pcal_c$ is an inductive set and by Zorn Lemma it has a maximal element that we write $P_c$. We take $\Kcal_c=P_c\Kcal$.

Suppose now that $P_c\ne0$. By definition, $\Phi(\Kcal_c)$ is a Commutative Subsystem of the Environment. Furthermore, if $\tilde{\Kcal}$ is a subspace of $\Kcal$ such that $\Phi(\tilde{\Kcal})$ is a Commutative Subsystem of the Environment, then the orthogonal projection on $\tilde{\Kcal}$ is dominated by $P_c$ so that $\tilde{\Kcal}$ is a subspace of $\Kcal_c$. Consequently $U^q_\cdot$ does not have any Commutative Subsystem.
\end{proof}

\begin{rem}We emphasize that Theorem \ref{theodecompcont} states the existence of a decomposition, without providing any practical way to explicit it in terms of the coefficients. Indeed, the first step in order to find the decomposition is the study the Environment Algebra $\Acal(\Sbb)$. For small matrices this can be done for instance numerically. For instance, in \cite{Bar2} it is proved that $\Acal(\Sbb)'$ is the eigenspace for the eigenvalue $1$ of a certain completely positive map on $\Kcal$. This provides by 
Theorem \ref{theodecompenv} a decomposition of $\Sbb$ as $\Sbb=\Sbb_1+\Sbb_2$, such that $\Sbb_1$ is the maximal block-diagonal unitary operator that we can extract from $\Sbb$. However this does not give the decomposition of $\Phi(\Kcal)$ into a classical and a quantum part.\\
In the next subsection we develop this point with several examples.
\end{rem}

\subsection{Examples and open problems}

In the first example below we prove that the spontaneous emission whose evolution is given by Equation \eqref{eqspontem} has a purely quantum environment.

\begin{ex}[Spontaneous Emission]\label{exnoncom}
Take $\Hcal=\C^2$ and $\Kcal=\C$. Let $U_\cdot$ be the solution of the quantum Langevin Equation
\begin{equation*}
dU_t=-\frac{1}{2}V^*V\,U_t\,dt+V\, U_t\,da^0_1(t)-V^*\,U_t\,da^1_0(t),
\end{equation*}
where 
$$
V=\begin{pmatrix}0&1\\0&0\end{pmatrix}\,.
$$
Here $\Sbb=I_\Hcal$, so that $\Acal(\Sbb)$ is commutative. However, clearly there does not exist $\lambda\in\C$ such that $V^*=\lambda V$, so that by Theorem \ref{theocomnoisealgebra}, $U_\cdot$ has a purely quantum environment.
\end{ex}

We rely on this typical evolution in order to construct one example where the decomposition is explicit.

\begin{ex}[An explicit decomposition]\label{exdecompexplicit}
Take $\Hcal=\C^2$ and $\Kcal=\C^2$. We consider the following coefficients in the quantum Langevin Equation, with $\lambda,\theta\in\R$, $\lambda>0$:
\begin{equation*}
\Sbb=\begin{pmatrix}
\sin^2\theta & \cos^2\theta & \sin\theta\cos\theta & \sin\theta\cos\theta \\
\cos^2\theta & \sin^2\theta &\sin\theta\cos\theta & \sin\theta\cos\theta \\
-\sin\theta\cos\theta & \sin\theta\cos\theta & -\cos^2\theta & \sin^2\theta \\
\sin\theta\cos\theta & -\sin\theta\cos\theta & \sin^2\theta & -\cos^2\theta
\end{pmatrix},
\end{equation*}
\begin{equation*}
L^0_1=\begin{pmatrix}
-\lambda\cos\theta & \lambda\cos\theta+\sin\theta \\
\lambda\cos\theta & -\lambda\cos\theta
\end{pmatrix},\qquad L^0_2=\begin{pmatrix}
-\lambda\sin\theta & \lambda\sin\theta-\cos\theta \\
\lambda\sin\theta & -\lambda\sin\theta
\end{pmatrix}.
\end{equation*}
After the change of noise given by the unitary operator
\[
W=\begin{pmatrix}
\sin\theta & \cos\theta \\ -\cos\theta & \sin\theta
\end{pmatrix},
\]
we find in the new basis:
\begin{equation*}
\tilde \Sbb=\begin{pmatrix}
1 & 0 & 0 & 0 \\
0 & 1 & 0 & 0 \\
0 & 0 & 0 & 1 \\
0 & 0 & 1 & 0
\end{pmatrix},\qquad \tilde L^0_1=\begin{pmatrix}
0 & 1 \\ 0 & 0
\end{pmatrix},\qquad \tilde L^0_2=\begin{pmatrix}
-\lambda & \lambda \\ \lambda & -\lambda
\end{pmatrix}=\lambda(S-I_\Hcal)\,,
\end{equation*}
where $S=\begin{pmatrix}
0 & 1 \\ 1 & 0
\end{pmatrix}$. Consequently the quantum Langevin Equation takes the form:
\[
dU_t=-\frac{1}{2}\left((\tilde L^0_1)^*\tilde L^0_1+(\tilde L^0_2)^2\right)U_t\,dt+\tilde L^0_1 U_t\,d\tilde a^0_1(t)-(\tilde L^0_1)^*U_t\,d\tilde a^1_0(t)+\tilde L^0_2 U_t\,dX_t,
\]
where $X_t$ is a compensated Poisson process of intensity $\lambda$ and jumps $1/\lambda$.

\smallskip
In this example, the decomposition between classical and quantum noises resumes to the decomposition of $\Sbb$, because it is possible to clearly identify the classical and quantum part after the diagonalization of $\Sbb$. This is not always possible.
\end{ex}

\begin{ex}[A non-explicit decomposition]\label{exdecompnonexplicit}
Take $\Hcal=\C^2$ and $\Kcal=\C^2$. We assume in this example that $\Sbb=I_{\Hcal\otimes\Kcal}$. Consequently the matrix of $\Sbb$ does not depend on the choice of the basis of $\Kcal$. However the coefficients $L^0_i$ do. Consider for instance:
\begin{equation*}
L^0_1=\begin{pmatrix}
-\cos\theta & \cos\theta+\sin\theta \\
\cos\theta & -\cos\theta
\end{pmatrix},\qquad L^0_2=\begin{pmatrix}
-\sin\theta & \sin\theta-\cos\theta \\
\sin\theta & -\sin\theta
\end{pmatrix}\,.
\end{equation*}
It happems that the same unitary operator $W$ as in the previous example gives the decomposition (with $\lambda=1$) in a classical and a quantum part. The point is that we know no algorithm in order to compute $W$, by only looking at the coefficients $L^0_1$ and $L^0_2$.
\end{ex}

\bibliographystyle{abbrv}
\bibliography{biblio}

\end{document}